%% file: main_ver3_after_reviewed.tex
\documentclass[conference]{IEEEtran}
\IEEEoverridecommandlockouts

\input{packages}
\input{acronyms}
\input{definitions_math}

\def\BibTeX{{\rm B\kern-.05em{\sc i\kern-.025em b}\kern-.08em
    T\kern-.1667em\lower.7ex\hbox{E}\kern-.125emX}}
\begin{document}

\title{Target Induced Angle Grid Regularized Estimation for Ghost Identification in Automotive Radar
}

\author{\author{Junho Kweon$^*$ and Vishal Monga$^*$\\
$^*$\textit{Electrical Engineering Department, The Pennsylvania State University}, University Park, PA, USA}
}

\maketitle

\begin{abstract}
This study presents a novel algorithm for identifying ghost targets in automotive radar by estimating complex valued signal strength across a two-dimensional angle grid defined by direction-of-arrival (DOA) and direction-of-departure (DOD). In real-world driving environments, radar signals often undergo multipath propagation due to reflections from surfaces such as guardrails. These indirect paths can produce ghost targets—false detections that appear at incorrect locations—posing challenges to autonomous navigation. A recent method, the Multi-Path Iterative Adaptive Approach (MP-IAA), addresses this by jointly estimating the DOA/DOD angle grid, identifying mismatches as indicators of ghost targets. However, its effectiveness declines in low signal-to-noise ratio (SNR) settings. To enhance robustness, we introduce a physics-inspired regularizer that captures structural patterns inherent to multipath propagation. This regularizer is incorporated into the estimation cost, forming a new loss function that guides our proposed algorithm, TIGRE (Target-Induced angle-Grid Regularized Estimation). TIGRE iteratively minimizes this regularized loss and we show that our proposed regularizer asymptotically enforces $\ell_0$ sparsity on the DOA/DOD grid.  Numerical experiments demonstrate that the proposed method substantially enhances the quality of angle-grid estimation across various multipath scenarios, particularly in low SNR environments, providing a more reliable basis for subsequent ghost target identification.
\end{abstract}


\section{Introduction}
\gls{mimo} radar is useful in estimating the target information thanks to the larger degree-of-freedom coming from multiple transmitters and receivers. With this advantage, \gls{mimo} radar is used in wide range of application such as automotive radar \cite{li2008mimo, sun2020mimo}. 

Radar can measure range, relative velocity, angle of azimuth, which can be used to understand surroundings and used for Advanced Driver Assistance Systems (ADAS). Automotive radar detect targets based on the signal received, and \gls{doa} estimation is one of the key technique widely used in this field. While sensing surrounding environment from measured radar signals, ghost target can appear when the received signals are not directly coming from targets \cite{ammen2020ghost, zheng2024detection, liu2020multipath}. Specifically, multipath caused by reflections by surroundings yields ghost target whose \gls{doa} and \gls{dod} of the radar signal are different \cite{bilik2019rise}. 
These ghost targets can burden in automotive radar operation by tracking all of actual target and ghost targets and predicting their movement \cite{ammen2020ghost, sun2020mimo}. Therefore, for safe and robust \gls{adas} or autonomous driving system, it is crucial to identify ghost targets. 

There have been some studies for ghost target recognition in radar. For example, the authors of \cite{feng2021multipath, fertig2017knowledge} exploit Doppler information to detect multipath. Specifically, \cite{feng2021multipath} used a linear relationship between the target and multipath ghosts once a clear range-Doppler map is acquired. However, if the target is stationary compared to the detecting automotive radar, Doppler information can be elusive. The study of \cite{liu2020multipath} suggests to regard the other detected targets in very similar range are ghost targets of a target. However, this approach is not applicable when they are located in different range or can fail to detect multiple targets in similar range. The study of \cite{zheng2024detection} used iterative method use the matrices whose dimensions are defined as the number of targets, so as it finds more targets, it iteratively extend the problem dimension. It is worth noting that the authors introduced the concept of group sparsity of reversibility of the first-order propagation path, which has one additional intermediate reflection compared to direct path that is the sequence of transmitter-target-receiver. 
However, the received signal amplitude for a first-order path and its reverse path may differ due to different phase accumulation (interference) pattern. The study of \cite{liu2024data} proposed data-driven method that can classify ghost targets when large-scale training data with annotations is available. In synthetic aperture radar system, the authors of \cite{setlur2011multipath} proposed a method that can track the ghost targets back to their actual target. 

The authors of \cite{li2022multipath} introduced \gls{mpiaa} estimating the \gls{doa}/\gls{dod} angle grid from the received signal and identifies targets with mismatched \gls{doa} and \gls{dod} as ghost targets. \gls{mpiaa} is based on the well-known IAA algorithm for estimating \gls{doa} \cite{yardibi2010source}. However, MP-IAA exhibits limited performance under low Signal-to-Noise Ratio (SNR) con-ditions, which are common in urban traffic scenarios. 

In this paper, we develop an iterative algorithm, termed \textit{\gls{tigre}}, to address \gls{doa}/\gls{dod} grid estimation. Our main contributions in this paper are
outlined as follows
\begin{itemize}
    \item \textbf{New Target Induced Regularizer:} We propose a physics-inspired regularizer that captures structural priors of ghost target inherent in multipath propagation. Specifically, we use the fact that a ghost target in the first-order path appears where either its \gls{doa} or \gls{dod} is the same with the actual target's angular position (\gls{doa}=\gls{dod}) while the other is different.

     \item \textbf{Asymptotic $\ell_0$ Sparsity in \gls{doa}=\gls{dod} Grids:} We prove that our proposed regularizer asymptotically (over the TIGRE iterations) enforces $\ell_0$ sparsity when \gls{doa}=\gls{dod}. Such sparsity leads to more accurate ghost target estimation as they are induced by the actual target's presence.
     
    \item \textbf{Custom Initialization:} We propose a special initialization for TIGRE that exploits the uncertainty in the existence of ghost targets and the ambiguity in their locations. Experimentally, we validate the improvements coming from the custom initialization in both performance and enhanced convergence speed.
\end{itemize}

\section{System Model}




Let $X_{g,q}\in\bbC$ is the signal magnitude and phase that is transmitted at $q$-th \gls{dod} and received at $g$-th \gls{doa}. Assuming the visible \gls{doa}/\gls{dod} are discretized into $G$ and $Q$ grids, we write the matrix version of $\{X_{g,q}\}_{g=1,q=1}^{G,Q}$ as
\begin{equation}
    \bX = 
    \begin{bmatrix}
        X_{1,1} & X_{1,2} & \cdots & X_{1,Q} \\
        X_{2,1} & X_{2,2} & \cdots & X_{2,Q} \\
        \vdots & \vdots & \ddots & \vdots \\
        X_{G,1} & X_{G,2} & \cdots & X_{G,Q}
    \end{bmatrix}\in\bbC^{G\times Q},
    \label{eq:X_define}
\end{equation}

\begin{figure}
    \centering
    \includegraphics[width=0.47\linewidth]{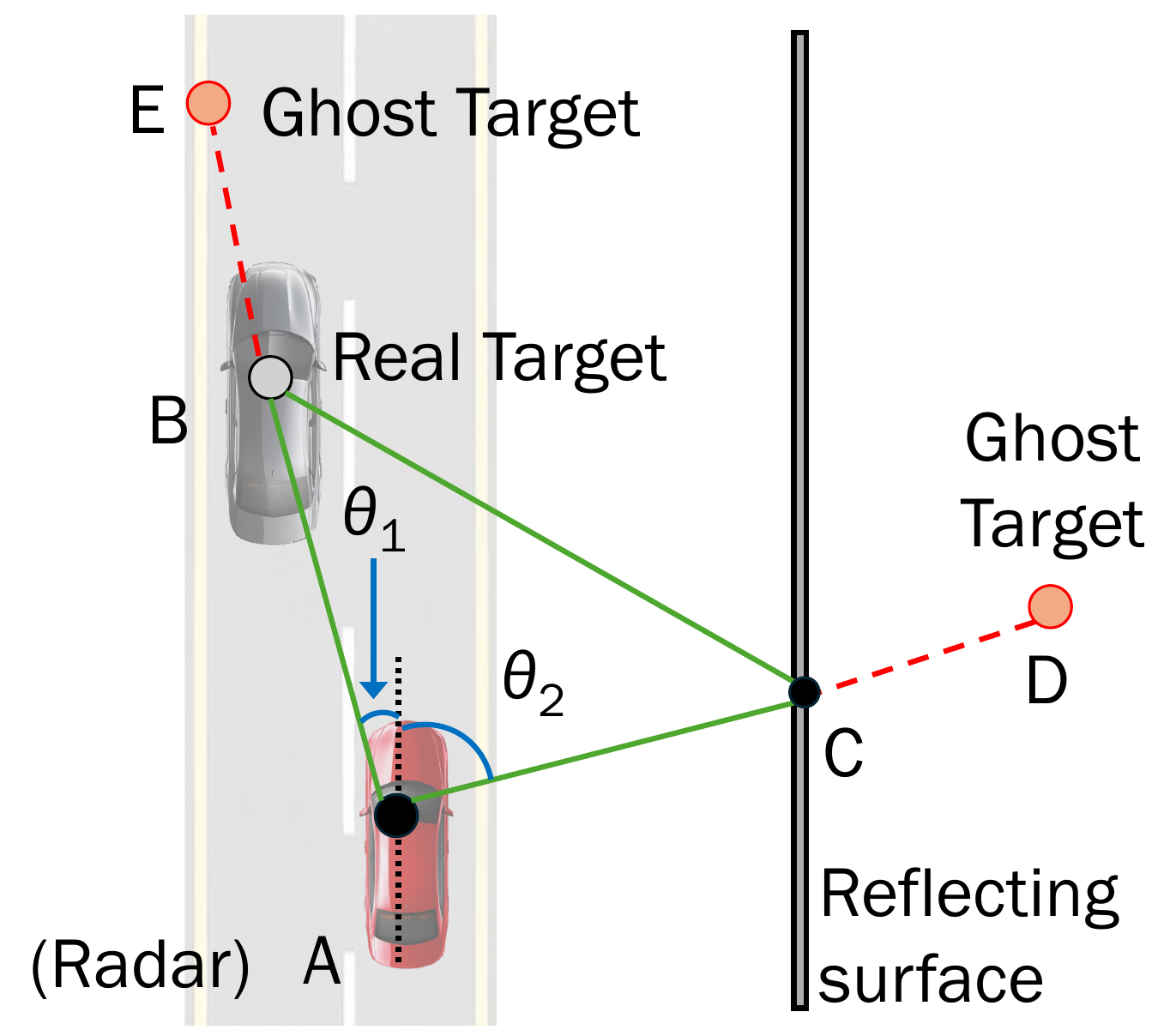}
    \vspace{-3pt}
    \caption{Geometry of multipath propagation in automotive radar. A is the automotive radar platform, B is real target, and C is reflection point.}
    \label{fig:automotive}
\end{figure}

Let the steering vectors for transmitter and the receiver at angle $\theta$ are $\ba_t(\theta)\in\bbC^{N_t\times 1}$ $\ba_r(\theta)\in\bbC^{N_r\times 1}$, which are
\begin{equation}
    \begin{aligned}
        \ba_t(\theta)=\left[1,e^{j2\pi\frac{d_t}{\lambda}\sin\theta}, ..., e^{j2\pi\frac{d_t}{\lambda}(N_t-1)\sin\theta}\right]^T\\
        \ba_r(\theta)=\left[1,e^{j2\pi\frac{d_r}{\lambda}\sin\theta}, ..., e^{j2\pi\frac{d_r}{\lambda}(N_r-1)\sin\theta}\right]^T.
    \end{aligned}    
    \label{eq:at_ar_define}
\end{equation}
We define the steering matrices for the discretized \gls{dod} and \gls{doa} as follows:
\begin{equation}
    \begin{aligned}
        \bA_t &\triangleq \left[ \ba_t(\theta_{t,1}), \ldots, \ba_t(\theta_{t,G}) \right] \in \mathbb{C}^{N_t \times G}\, \\
        \bA_r &\triangleq \left[ \ba_r(\theta_{r,1}), \ldots, \ba_r(\theta_{r,Q}) \right] \in \mathbb{C}^{N_r \times Q},
    \end{aligned} 
    \label{eq:At_Ar_define}
\end{equation}
and define the following matrix $\bA$:
\begin{equation}
    \bA=\bA_r\otimes\bA_t=[\ba_1, \ba_2, \cdots, \ba_{GQ}]\in\bbC^{NtNr\times GQ},
\end{equation}
where $\otimes$ is Kronecker product and $\ba_i\in\bbC^{N_tN_r\times 1}\ \forall i$.
Using these matrices, the received signal can be expressed as
\begin{equation}
    \by = \vectrz\{\bA_r\bX\bA_t^T\} + \be\in\bbC^{N_tN_r\times 1}
    = \bA\bx + \be
    \label{eq:y_define}
\end{equation}
where $\bx\in\bbC^{GQ\times 1}$ is a vectorized version of $\bX$, \ie, 
\begin{equation}
    \bx = [X_{1,1},X_{2,1},\cdots,X_{G,1},X_{1,2},X_{2,2}\cdots,X_{G,Q}]^T,
    \label{eq:x_define}
\end{equation}
and $\be$ is noise. 

In automotive radar, only direct paths and first-order paths are considered in general as higher-order paths naturally have more severe path loss resulted from longer distance or can be easily distinguished by different distance comparatively \cite{zheng2024detection}. The recent study of \cite{li2022multipath} proposed \gls{mpiaa} to estimate \gls{doa}/\gls{dod} angle grid to identify ghost targets. The method aims to solve the following element-wise problem:
\begin{equation}
\begin{split}
    \underset{X_{g,q}^n(=x_i)}{\min}\ f_{g,q}(X_{g,q}&)    
    =  \underset{x_i}{\min}\ 
    \left[\by-x_i\ba_i\right]^H\bQ_i^{-1}(\bx)\left[\by-x_i\ba_i\right]\\
    = & \underset{X_{g,q}}{\min}\ 
    \left[\by-X_{g,q}\ba_i\right]^H\bQ_i^{-1}(\bx)\left[\by-X_{g,q}\ba_i\right]
\end{split}
    \label{eq:WLS_original}
\end{equation}
where $i=g + (q-1)G$ so that $X_{g,q}=x_i$ in \eqref{eq:x_define},
\begin{equation}
    \ba_i = \ba_r(\theta_{r,g})\otimes\ba_t(\theta_{t,q}),
\end{equation}
and $\bQ_i(\bx)$ is noise covariance matrix such that
\begin{equation}
    \bQ_i(\bx)=\bA\diag(|\bx|^2)\bA^H - x_i\ba_i\ba_i^H
    \in\bbC^{N_tN_r\times N_tN_r}.
    \label{eq:Q_i}
\end{equation}
Therefore, the problem is nonconvex where global minimum is elusive \cite{monga2017handbook}. To address this minimization problem, \gls{mpiaa} updates $X_{g,q}$ by using a closed form solution of the \gls{wls} problem in \eqref{eq:WLS_original} with assuming $\bQ_i(\bx)$ is fixed and computes $\bQ_i(\bx)$ separately. Once correctly reconstructed from the received signal, $X_{g,q}$ with significant amplitude ($|X_{g,q}|\gg0$) is regarded as \textit{detected target}. The detected targets with \gls{doa}$\neq$\gls{dod} can be regarded as a ghost target \cite{li2022multipath}.


\section{Proposed Method}
\subsection{Target Induced Regularizer on DOA$\neq$DOD Grids}
\label{subsec:gs_regularizer}
Without losing generalizability, we set $G=Q$ so that 
the diagonal elements of $\bX$ becomes \gls{doa}$=$\gls{dod} case\footnotemark.
\footnotetext{If $G\neq Q$, one can still use the proposed method by applying the \gls{doa}$=$\gls{dod} case to the corresponding elements of $\bX$ in their setup.} From \eqref{eq:WLS_original}, we can deduce that the generalized cost to be minimized can be expressed as $\Tilde{f}(\bX)=\sum_{g=1,q=1}^{G,Q} f_{g,q}(X_{g,q})$. 
At $n$-th iteration step to acquire $\bX^{n+1}$, we introduce the regularizer $\tcR^{n}(\bX)=\sum_{g=1,q=1}^{G,Q}\cR_{g,q}^{n}(X_{g,q})$ where
\begin{equation}
    \cR_{g,q}^{n}(X_{g,q}) 
    = \mu_{g,q}^n|X_{g,q}|^2
    = \frac{1}{|X_{g,g}^n|^2+|X_{q,q}^n|^2+\epsilon_0}|X_{g,q}|^2,
    \label{eq:regularizer}
\end{equation}
yielding the loss function 
\begin{equation}
\begin{split}
    \cL^{n}(\bX)&=\sum_{g=1,q=1}^{G,Q} \cL_{g,q}^{n}(X_{g,q})\\
    &= \sum_{g=1,q=1}^{G,Q} f_{g,q}^{n}(X_{g,q}) + \lambda_{g,q}\cR_{g,q}^{n}(X_{g,q}),
\end{split}
\end{equation}
where $\lambda_{g,q}>0$ controls the importance of the regularizer $\cR_{g,q}^n(X_{g,q})\ \forall g,q,n$.

In \eqref{eq:regularizer}, $\epsilon_0>0$ is a small positive number to prevent the denominator become zero. When either $|X_{g,g}|^n$ or $|X_{q,q}^n|$ is small, $\mu_{g,q}^n>0$ becomes large, thereby assigning greater weight to $|X_{g,q}|^2$ in the minimization of $\cR_{g,q}^n(X_{g,q})$.
Therefore, this regularizer encourages that the estimated \gls{doa}/\gls{dod} grids $X_{g,q}\ (g\neq q)$ have a ghost target appear only when there is an actual target either $(g,g)$ or $(q,q)$.

\textit{\acrfull{tigre}} aims to solve the following element-wise (or grid-wise) loss function to acquire $X_{g,q}^{n+1}$:

\begin{equation}
\begin{split}
    X_{g,q}^{n+1}
    =&\argmin_{X_{g,q}}\,\cL_{g,q}^{n}(X_{g,q})\\
    =&\argmin_{X_{g,q}}\,f_{g,q}^{n}(X_{g,q}) + \lambda_{g,q}\cR_{g,q}^{n}(X_{g,q})
\end{split}
    \label{eq:X_n_plus_1}
\end{equation}
 

We can derive the closed form solution of \eqref{eq:X_n_plus_1} by differentiating $\cL_{g,q}^{n}(X_{g,q})$ with $X_{g,q}$:
\begin{equation}
    \begin{split}
        \frac{\partial}{\partial X_{g,q}^*}&\cL_{g,q}^n(X_{g,q})
        = \frac{\partial}{\partial X_{g,q}^*} \left\{f_{g,q}^{n}(X_{g,q}) + \lambda_{g,q} \mathcal{R}_{g,q}^{n}(X_{g,q})\right\}\\
        =& -\ba_i^H\bW_i^{n}\left(\by-X_{g,q}\ba_i\right) + \frac{\lambda_{g,q}X_{g,q}}{|X_{g,g}^n|^2 + |X_{q,q}^n|^2 + \epsilon_0},
    \end{split}
    \label{eq:partial_derivative}
\end{equation}
where $\bW_i^{n}=\bQ_i^{-1}(\bx^n)$.
Defining $D_{g,q}^n=|X_{g,g}^n|^2 + |X_{q,q}^n|^2 + \epsilon_0$ for brevity, we can find a solution that will be in a stationary point of $\cL_{g,q}^{n}(X_{g,q})$ in terms of $X_{g,q}$:
\begin{equation}
    -D_{g,q}^n\ba_i^H\bQ_i^{-1}\left(\by-X_{g,q}\ba_i\right) + \lambda_{g,q}X_{g,q} = 0
\end{equation}

We can find a closed form solution of $X_{g,q}$ as
\begin{equation}
    X_{g,q}^{n+1} = \frac{D_{g,q}^n(\ba_i^H\bW_i^{n}\by)}{D_{g,q}^n(\ba_i^H\bW_i^{n}\ba_i) + \lambda_{g,q}}.
    \label{eq:X_gq_star}
\end{equation}

\subsection{Asymptotic $\ell_0$ Sparsity Behavior in DOA$/$DOD Grids}
\label{subsec:diag}
Regarding the effect of the proposed regularizer to the diagonal part of $\bX$, we show the following lemma:
\begin{lemma}
    The proposed regularizer in \eqref{eq:regularizer} encourages $\ell_0$ sparsity asymptotically in the \gls{doa}$=$\gls{dod} grids of $\bX$, \ie, the diagonal elements of $\bX$ when $G=Q$.
\end{lemma}

\begin{proof}
    Let $\bz=\diag(\bX)\in\bbC^{G\times 1}$ and $\tcR_\diag^{n}(\bz)=\sum_{g=1}^{G}\cR_{g,g}^{n}(z_{g})=\sum_{g=1}^{G}\cR_{g,g}^{n}(X_{g,g})$. 
    Assuming $z_g^n$ converges to $\brz_g$ as $n\rightarrow \infty$, the regularizer $\cR_{g,g}^n(z_g^{n+1})$ becomes
    \begin{equation}
        \lim_{n\rightarrow \infty}\cR_{g,g}^{n}(z_g^{n+1}) 
        = \lim_{z_g^n\rightarrow \brz_g}\frac{|z_g^{n+1}|^2}{2|z_g^n|^2+\epsilon_0}
        \label{eq:Rgg_conv}
    \end{equation}
    As the above term is approximately $1/2$ when $\brz_g\neq 0$ and $0$ when $\brz_g=0$ \cite{candes2008enhancing}, we can approximate
    \begin{equation}
        \lim_{n\rightarrow \infty}\tcR_\diag^{n}(\bz^{n+1})
        =\lim_{n\rightarrow \infty}\sum_{g=1}^G\cR_{g,g}^{n}(z_g^{n+1})
        \approx \frac{1}{2}\|\bz\|_0,
        \label{eq:R_l0norm}
    \end{equation}
    \ie, regularizer $\tcR_\diag^{n}(\bz)$ is approximated as $\ell_0$ measure of $\bz$.
    This naturally enforces sparsity on \gls{doa}$=$\gls{dod}, i.e. the diagonal of $\bX$. Because a $0$ in the diagonal also encourages the corresponding row and column to be all $0$'s following \eqref{eq:regularizer} - $\bf X$ is also expected to be sparse asymptotically.
\end{proof}

\subsection{Custom Initialization}
\label{subsec:init}
As the existence of multipath is unknown in the beginning, we aim to find initial guess of $\bX$ by finding the best diagonal matrix for a given measured signal $\by$. Under this setting, the initial $\bX$ can be found via solving the following:
\begin{equation}
    \min_{\mathbf{X} = \mathrm{diag}(\mathbf{z})} \big\| \by - \vectrz\big(\mathbf{A}_r \mathbf{X} \mathbf{A}_t^T\big) \big\|_2^2,
    \label{eq:initial_guess}
\end{equation}
where $\bz=[z_1,z_2\cdots,z_G]^T\bbC^{G\times 1}$ such that $X_{g,g}=z_g$.

Define the matrix $\mathbf{C} \in \mathbb{C}^{N_tN_r \times G}$ as:
\begin{equation}
    \mathbf{C} = [ \mathbf{a}_{r,1} \otimes \mathbf{a}_{t,1}, \mathbf{a}_{r,2} \otimes \mathbf{a}_{t,2}, \cdots , \mathbf{a}_{r,G} \otimes \mathbf{a}_{t,G} ],
\end{equation}
where $\ba_{r,g}=\ba_r(\theta_{r,g})$ and $\ba_{t,g}=\ba_t(\theta_{t,q})$.
Then the optimization problem becomes:
\begin{equation}
    \min_{\mathbf{z} \in \bbC^{G\times 1}} \left\| \by - \mathbf{C} \bz \right\|_2^2
\end{equation}
The closed-form solution is:
\begin{equation}
    \mathbf{z}^\star = \big(\mathbf{C}^H \mathbf{C}\big)^{-1} \mathbf{C}^H \by
    \label{eq:z_star}
\end{equation}
Finally, the estimated diagonal matrix is:
\begin{equation}
    \mathbf{X}^0 = \diag(\mathbf{z}^\star)
    \label{eq:X_initial}
\end{equation}

\subsection{Step-By-Step Algorithm of TIGRE}
\cref{algorithm:TIGRE} shows the step-by-step algorithm to estimate \gls{doa}/\gls{dod} angle grid $\bX$ from measured signal $\by$. In line 3, using the custom initialization, $\bX$ is initialized as $\bX^0$. In line 5-7, $\bX$ is updated aiming to minimize the loss defined in \eqref{eq:X_n_plus_1} until convergence of $\bX$.

\begin{algorithm}[t]
    \caption{TIGRE algorithm.}
    \label{algorithm:TIGRE}
    \begin{algorithmic}[1]
        
        \State {\textbf{Inputs}: $\by,\ G,\ Q,\ \epsilon_0,\ \epsilon_\bx,\ \{\lambda_{g,q}\}_{g=1,q=1}^{G,Q},\ N_{\max}$ } 
        \State \textbf{Outputs}: $\bX^\star$
        \State Initialize $\bX^0=\diag(\mathbf{x}^\star)$ following \eqref{eq:z_star} and \eqref{eq:X_initial}.
        \State \textbf{for} $n=0:N_{\max}$
        \State\quad Calculate $\bQ_i^n(\bx^n)\ \forall i$ following \eqref{eq:Q_i} for $\bx^n=\vectrz(\bX^n)$.\vspace{2pt}
        \State \quad Compute $X_{g,q}^{\star, n+1}$ using \eqref{eq:X_gq_star} for $\{g,q\}_{g=1,q=1}^{G,Q}$.
        \vspace{2pt}
        \State \quad $\bx^{n+1}=\vectrz\big(\bX^{\star, n+1}\big)$
        \State \quad \textbf{if} $|\bx^{n+1}-\bx^n|<\epsilon_\bx$ \textbf{then}
        \State \qquad break
        \State \quad \textbf{end if}
        \State \textbf{end for}
        \State $\bX^\star = \bX^{\star, n+1}$
    \end{algorithmic}
\end{algorithm}

\section{Numerical Results}
\subsection{Simulation Settings}
We consider co-located \gls{mimo} radar with $N_t=N_r=8$ and $d_t=d_r=\lambda/2$ where $\lambda$ is wavelength. \gls{snr} is set to 10 dB. $\epsilon_\bx=10^{-2}$ to decide convergence in $\bx$ during iteration of the competing optimization methods.

\begin{table}[t]
    \caption{DOA/DOD angular information.}
    \label{tab:target_info}
    \centering
    \begin{tabular}{c||c|c|c|c|c|c}
         Target & \multicolumn{2}{c|}{\textbf{Actual Target}}  & \multicolumn{2}{c|}{\textbf{Ghost 1}} & \multicolumn{2}{c}{\textbf{Ghost 2}}\\
         \cline{2-7} 
         No. & $|X|$ & $\theta_r=\theta_t$ 
         & $|X|$ & $(\theta_r, \theta_t)$ 
         & $|X|$ & $(\theta_r, \theta_t)$ \\
        \hline
        1 
        & $1$ & $-20^\circ$ 
        & $0.7$ & $(-20,40)^\circ$ 
        & $0.5$ & $(40,-20)^\circ$ \\
        2
        & $1$ & $-60^\circ$ 
        & $0.7$ & $(-60,60)^\circ$ 
        & $0.5$ & $(60,-60)^\circ$ \\
        3
        & $1$ & $-40^\circ$ 
        & $0.7$ & $(-40,50)^\circ$
        & $0.5$ & $(50,-40)^\circ$ \\
    \end{tabular}
    \vspace{-10pt}
\end{table}

We compare \gls{tigre} with state-of-the-art \gls{ira} \cite{fang2016super, tang2021off} and
\gls{mpiaa} \cite{li2022multipath}\footnotemark.
\footnotetext{\gls{ira} is a method for \gls{doa} estimation, but adapted to the \gls{doa}/\gls{dod} grid.} 
To evaluate the performance gain of the custom initialization, we simulated both \gls{tigre} with random initialization of $\bX$--marked as \gls{tigre}-no Init.--and the custom initialization articulated in \cref{subsec:init}--denoted as \gls{tigre}. For hyperparameters of \gls{tigre}, we set $\lambda_{g,q}=10$ for $g\neq q$ and $\lambda_{g,q}=1$ for $g=q$ via cross-validation \cite{monga2017handbook}.

\subsection{Simulation Results}
\label{subsec:results}
\cref{tab:result_target1} shows numerical results of the \gls{tigre} and competing methods for single target scenario: target 1 in \cref{tab:target_info}. The values (Error $\|\bX^\star-\bX^\text{true}\|_F^2$, number of iterations, and computational time in seconds) are average of 100 random samples for fixed target information, and the best two performance is bold faced. Both \gls{tigre}-no Init. and \gls{tigre} shows the smallest estimation error among the methods. \gls{tigre} shows the second smallest number of iterations and shortest computation time, but very close to the best ones by \gls{ira}. 


\begin{table}
    \small
    \caption{Averaged results for one target.}
    \label{tab:result_target1}
    \centering
    \begin{tabular}{c||c|c|c}
        \textbf{Method} & \textbf{Error} $(\downarrow)$ & \textbf{Iteration} $(\downarrow)$ & \textbf{Time (sec)} $(\downarrow)$\\
        \hline
        \gls{ira} \cite{fang2016super, tang2021off} 
        & 4.91 & \textbf{19.53} & \textbf{1.91}\\
        \gls{mpiaa} \cite{li2022multipath}
        & 4.58 & 36.59 & 3.25 \\
        \gls{tigre}-no Init.  
        & \textbf{3.06} & 22.45 & 2.31  \\
        \gls{tigre} 
        & \textbf{0.59} & \textbf{20.07} & \textbf{2.05}
    \end{tabular}
\end{table}

\cref{tab:result_target2} shows the results when there are two targets: target 1 and target 2 in \cref{tab:target_info}. \gls{tigre} shows its superiority in the estimation accuracy and the computational efficiency. The visual results shown in \cref{fig:X_2D} (a), (b) are visualization of resulting \gls{doa}/\gls{dod} angle grid $\bX$ in this two target case, where the color denotes the value of $|\bX|$. While \gls{mpiaa} and \gls{tigre} show the best accuracy in \cref{tab:result_target2}, the result of \gls{mpiaa} has multiple peaks in its diagonal part (at the location of the actual target)
like $(30,30)^\circ$ and $(80,80)^\circ$. 

\begin{table}[t]
    \caption{Results for two targets.}
    \label{tab:result_target2}
    \centering
    \begin{tabular}{c||c|c|c}
        \textbf{Method} & \textbf{Error} $(\downarrow)$ & \textbf{Iteration} $(\downarrow)$ & \textbf{Time (sec)} $(\downarrow)$\\
        \hline
        \gls{ira} \cite{fang2016super, tang2021off} 
        & 5.52 & \textbf{14.08} & \textbf{1.26}\\
        MP-IAA \cite{li2022multipath}
        & 4.69 & 27.03 & 2.29 \\
        TIGRE-no Init.  
        & \textbf{3.14} & 17.82 & 1.83  \\
        \gls{tigre} 
        & \textbf{0.86} & \textbf{15.50} & \textbf{1.59}
    \end{tabular}
\end{table}

\cref{tab:result_target3} shows the case of three targets: target 1, 2, and 3. Again, \gls{tigre} shows the outstanding performance among the counterparts. \cref{fig:X_2D} (c), (d) show that, while \gls{mpiaa} has multiple false alarms, \gls{tigre} produces more precisely estimated \gls{doa}/\gls{dod} angle grid. \cref{fig:x_1D} shows the 1D-plot of resulting amplitude of $\bx(=\vectrz(\bX))$ and the points that is larger than 0.4 is marked. Green solid stems (with thicker line) shows true values, \ie, having peaks at the three actual target locations (three peaks with the largest values) and two ghost targets of each target. Blue solid line (thinner) is the result from \gls{mpiaa}, and red dotted line is from \gls{tigre}. While the \gls{mpiaa} yields a lot of false alarms (marked as $+$), \gls{tigre} provides comparatively correct estimation (marked as $\times$), especially capturing the actual targets.

In Table \ref{tab:result_target1}-\ref{tab:result_target3}, the benefits of \gls{tigre} over \gls{tigre}-no Init. in the estimation performance and convergence speed are enabled by the custom initialization proposed in \cref{subsec:init}.

\begin{table}[t]
    \caption{Results for three targets.}
    \label{tab:result_target3}
    \centering
    \begin{tabular}{c||c|c|c}
        \textbf{Method} & \textbf{Error} $(\downarrow)$ & \textbf{Iteration} $(\downarrow)$ & \textbf{Time (sec)} $(\downarrow)$\\
        \hline
        \gls{ira} \cite{fang2016super, tang2021off} 
        & 6.71 & \textbf{20.68} & \textbf{1.44}\\
        MP-IAA \cite{li2022multipath}
        & 5.62 & 28.25 & 1.94 \\
        TIGRE-no Init.  
        & \textbf{3.39} & 21.58 & 1.54  \\
        TIGRE 
        & \textbf{1.57} & \textbf{13.00} & \textbf{1.20}
    \end{tabular}
\end{table}
    
\section{Conclusion}
    WE develop \gls{tigre} - a regularized estimator over the \gls{doa}/\gls{dod} grid for ghost target identification in automotive radar. Unlike existing methods that do not consider the angular correlation between an actual target and its ghost targets, \gls{tigre} develops target induced regularizers to mitigate the estimated signal at \gls{doa}/\gls{dod} where a ghost target cannot exist based on the estimated actual target locations. Further, custom initialization is designed that finds the best \gls{doa}/\gls{dod} estimation based on only actual targets, which is prior to ghost targets. \gls{tigre} is shown to outperform the state of the art in both estimation accuracy and computational efficiency. 

\begin{figure}
    \centering
    \includegraphics[width=0.85\linewidth]{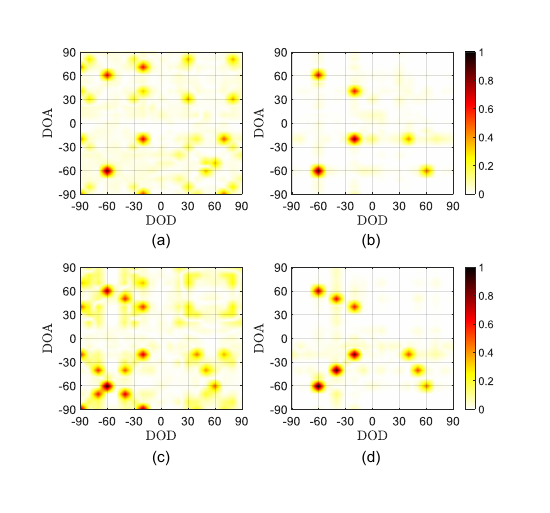}
    \vspace{-3pt}
    \caption{Estimated \gls{doa}/\gls{dod} angle grid $\bX$. (a), (b) are the case of two targets; and (c), (d) are the case of three targets. (a), (c) are from \gls{mpiaa}; and (b),(d) are from the proposed \gls{tigre} algorithm.}
    \label{fig:X_2D}
\end{figure}
\begin{figure}
    \centering
    \includegraphics[width=0.7\linewidth]{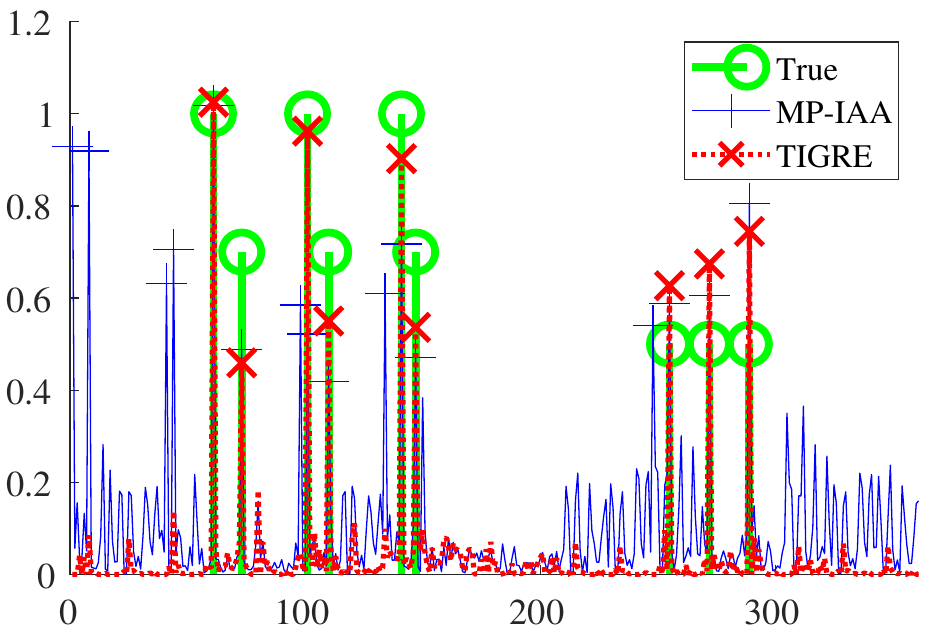}
    \vspace{-3pt}
    \caption{Angle grids in a vector form $\bx$ estimated by MP-IAA and TIGRE for the scenario of three targets, each having two ghost targets.}
    \label{fig:x_1D}
\end{figure}




\bibliographystyle{IEEEtran}

\bibliography{reference}

\end{document}

%% file: packages.tex
\usepackage{cite}
\usepackage{amsmath,amssymb,amsfonts}
\usepackage{algorithm}
\usepackage{algpseudocode}
\usepackage{graphicx}
\usepackage{textcomp}
\usepackage{xcolor}
\usepackage[acronyms]{glossaries}
\usepackage{xspace}
\usepackage{tabularx}
\usepackage{colortbl}
\usepackage[breaklinks=true,colorlinks,bookmarks=false]{hyperref}
\usepackage[capitalize]{cleveref}
\usepackage{enumitem}
\usepackage{multirow}
\usepackage{amsthm}
\makeatletter
\newcommand{\Eqref}[1]{\textup{\tagform@{\ref*{#1}}}}
\makeatother

%% file: acronyms.tex
\newacronym{adas}{ADAS}{Advanced Driver-Assistance Systems}
\newacronym{blrc}{BLRC}{Bayesian Linear Regression with Cauchy Prior}
\newacronym{doa}{DOA}{Direction-Of-Arrival}
\newacronym{dod}{DOD}{Direction-Of-Departure}
\newacronym{iaa}{IAA}{Iterative Adaptive Approach}
\newacronym{ira}{IRA}{Iterative Reweighted Algorithm}
\newacronym{mimo}{MIMO}{Multiple-Input, Multiple-Output}
\newacronym{mpiaa}{MP-IAA}{Multi-Path IAA}
\newacronym{snr}{SNR}{Signal-to-Noise Ratio}
\newacronym{tigre}{TIGRE}{Target Induced angle-Grid Regularized Estimation}
\newacronym{wls}{WLS}{Weighted Least-Squares}

%% file: definitions_math.tex
\def\bbC{\ensuremath{{\mathbb C}}}

\def\ba{\ensuremath{{\bf a}}}

\def\be{\ensuremath{{\bf e}}}

\def\bx{\ensuremath{{\bf x}}}

\def\by{\ensuremath{{\bf y}}}
\def\bz{\ensuremath{{\bf z}}}
\def\brz{\ensuremath{\bar{z}}}

\def\bA{\ensuremath{{\bf A}}}

\def\cL{\ensuremath{{\mathcal L}}}

\def\bQ{\ensuremath{{\bf Q}}}

\def\bW{\ensuremath{{\bf W}}}
\def\bX{\ensuremath{{\bf X}}}

\def\cL{\ensuremath{{\mathcal L}}}

\def\cR{\ensuremath{{\mathcal R}}}
\def\tcR{\ensuremath{\widetilde{\mathcal R}}}

\def\diag{\ensuremath{{\operatorname{diag}}}}

\def\vectrz{\ensuremath{{\mathrm{vec}}}}

\makeatletter
\providecommand{\onedot}{\futurelet\@let@token\@onedot}
\providecommand{\@onedot}{\ifx\@let@token.\else.\null\fi\xspace}
\makeatother

\providecommand{\ie}{\emph{i.e}\onedot}

\DeclareMathOperator*{\argmin}{arg\,min}

\makeatletter
\newcommand*{\addFileDependency}[1]{
  \typeout{(#1)}
  \@addtofilelist{#1}
  \IfFileExists{#1}{}{\typeout{No file #1.}}
}
\makeatother

\newtheorem{theorem}{Theorem}
\newtheorem{lemma}[theorem]{Lemma}
\crefname{lemma}{lemma}{lemmas}

%% file: reference.bib
@inproceedings{li2022multipath,
  title={Multipath ghost target identification for automotive MIMO radar},
  author={Li, Yunda and Shang, Xiaolei},
  booktitle={2022 IEEE 96th Vehicular Technology Conference (VTC2022-Fall)},
  pages={1--5},
  year={2022},
  organization={IEEE}
}

@article{yardibi2010source,
  title={Source localization and sensing: A nonparametric iterative adaptive approach based on weighted least squares},
  author={Yardibi, Tarik and Li, Jian and Stoica, Petre and Xue, Ming and Baggeroer, Arthur B},
  journal={IEEE Transactions on Aerospace and Electronic Systems},
  volume={46},
  number={1},
  pages={425--443},
  year={2010},
  publisher={IEEE}
}

@article{liu2024data,
  title={A Data-Driven Method for Indoor Radar Ghost Recognition with Environmental Mapping},
  author={Liu, Ruizhi and Song, Xinghui and Qian, Jiawei and Hao, Shuai and Lin, Yue and Xu, Hongtao},
  journal={IEEE Transactions on Radar Systems},
  year={2024},
  publisher={IEEE}
}

@article{feng2021multipath,
  title={Multipath ghost recognition for indoor MIMO radar},
  author={Feng, Ruoyu and De Greef, Eddy and Rykunov, Maxim and Sahli, Hichem and Pollin, Sofie and Bourdoux, Andr{\'e}},
  journal={IEEE Transactions on Geoscience and Remote Sensing},
  volume={60},
  pages={1--10},
  year={2021},
  publisher={IEEE}
}

@article{setlur2011multipath,
  title={Multipath model and exploitation in through-the-wall and urban radar sensing},
  author={Setlur, Pawan and Amin, Moeness and Ahmad, Fauzia},
  journal={IEEE Transactions on Geoscience and Remote Sensing},
  volume={49},
  number={10},
  pages={4021--4034},
  year={2011},
  publisher={IEEE}
}

@article{zheng2024detection,
  title={Detection of ghost targets for automotive radar in the presence of multipath},
  author={Zheng, Le and Long, Jiamin and Lops, Marco and Liu, Fan and Hu, Xueyao and Zhao, Chuanhao},
  journal={IEEE Transactions on Signal Processing},
  year={2024},
  publisher={IEEE}
}

@article{sun2020mimo,
  title={MIMO radar for advanced driver-assistance systems and autonomous driving: Advantages and challenges},
  author={Sun, Shunqiao and Petropulu, Athina P and Poor, H Vincent},
  journal={IEEE Signal Processing Magazine},
  volume={37},
  number={4},
  pages={98--117},
  year={2020},
  publisher={IEEE}
}

@book{li2008mimo,
  title={MIMO radar signal processing},
  author={Li, Jian and Stoica, Petre},
  year={2008},
  publisher={John Wiley \& Sons}
}

@inproceedings{liu2020multipath,
  title={Multipath propagation analysis and ghost target removal for FMCW automotive radars},
  author={Liu, Chenwen and Liu, Shengheng and Zhang, Cheng and Huang, Yongming and Wang, Haiming},
  booktitle={IET International Radar Conference (IET IRC 2020)},
  volume={2020},
  pages={330--334},
  year={2020},
  organization={IET}
}

@book{monga2017handbook,
  title={Handbook of convex optimization methods in imaging science},
  author={Monga, Vishal},
  year={2017},
  publisher={Springer}
}

@article{tang2021off,
  title={Off-grid DOA estimation with mutual coupling via block log-sum minimization and iterative gradient descent},
  author={Tang, Wen-Gen and Jiang, Hong and Zhang, Qi},
  journal={IEEE Wireless Communications Letters},
  volume={11},
  number={2},
  pages={343--347},
  year={2021},
  publisher={IEEE}
}

@article{fang2016super,
  title={Super-resolution compressed sensing for line spectral estimation: An iterative reweighted approach},
  author={Fang, Jun and Wang, Feiyu and Shen, Yanning and Li, Hongbin and Blum, Rick S},
  journal={IEEE Transactions on Signal Processing},
  volume={64},
  number={18},
  pages={4649--4662},
  year={2016},
  publisher={IEEE}
}

@inproceedings{ammen2020ghost,
  title={A ghost target suppression technique using interference replica for automotive FMCW radars},
  author={Ammen, Daiki and Umehira, Masahiro and Wang, Xiaoyan and Takeda, Shigeki and Kuroda, Hiroshi},
  booktitle={2020 IEEE Radar Conference (RadarConf20)},
  pages={1--5},
  year={2020},
  organization={IEEE}
}

@article{fertig2017knowledge,
  title={Knowledge-aided processing for multipath exploitation radar (MER)},
  author={Fertig, Louis B and Baden, J Michael and Guerci, Joseph R},
  journal={IEEE Aerospace and Electronic Systems Magazine},
  volume={32},
  number={10},
  pages={24--36},
  year={2017},
  publisher={IEEE}
}

@article{bilik2019rise,
  title={The rise of radar for autonomous vehicles: Signal processing solutions and future research directions},
  author={Bilik, Igal and Longman, Oren and Villeval, Shahar and Tabrikian, Joseph},
  journal={IEEE signal processing Magazine},
  volume={36},
  number={5},
  pages={20--31},
  year={2019},
  publisher={IEEE}
}

@article{candes2008enhancing,
  title={Enhancing sparsity by reweighted $l_1$ minimization},
  author={Candes, Emmanuel J and Wakin, Michael B and Boyd, Stephen P},
  journal={Journal of Fourier analysis and applications},
  volume={14},
  pages={877--905},
  year={2008},
  publisher={Springer}
}
